\documentclass{comnettoarxiv}

\usepackage{amsmath}  
\usepackage{amsfonts} 
\usepackage{amsthm} 
\usepackage{graphicx} 
\usepackage{float}

\usepackage{algorithmicx}
\usepackage{algpseudocode}
\usepackage{algorithm}

\usepackage{blkarray}
\usepackage{subcaption}

\usepackage[table]{xcolor}

\usepackage[cal=boondoxo]{mathalfa}

\usepackage[font=small,labelfont=bf]{caption}

\DeclareMathAlphabet{\mathcal}{OMS}{cmsy}{m}{n}

\newcommand{\cD}{{\cal D}}
\newcommand{\cE}{{\cal E}}
\newcommand{\cF}{{\cal F}}
\newcommand{\cG}{{\cal G}}
\newcommand{\cH}{{\cal H}}
\newcommand{\cL}{{\cal L}}
\newcommand{\cM}{{\cal M}}
\newcommand{\cN}{{\cal N}}

\newcommand{\cP}{{\cal P}}
\newcommand{\cS}{{\cal S}}
\newcommand{\cT}{{\cal T}}
\newcommand{\cV}{{\cal V}}
\newcommand{\cW}{{\cal W}}

\newcommand*\samethanks[1][\value{footnote}]{\footnotemark[#1]}

\newcommand{\comps}[1]{\mbox{comps}\left( #1 \right)}

\begin{document}

\title{Scalable Edge Clustering of Dynamic Graphs via Weighted Line Graphs}

\shorttitle{Dynamic Graph Edge Clustering} 
\shortauthorlist{Ostroski, et al.} 

\author{
\name{Michael Ostroski}
\address{National Security Agency\\
Fort Meade, MD 20755\\
maostro@uwe.nsa.gov
}
\and
\name{Geoffrey Sanders\thanks{
LLNL-JRNL-857155}, 
Trevor Steil\samethanks, Roger Pearce\samethanks}
\address{Lawrence Livermore National Laboratory (LLNL) \\
Livermore, CA, USA \\
\{sanders29, steil1, pearce7\}@llnl.gov
}
}

\maketitle

\begin{abstract}
{ 
\qquad Timestamped relational datasets consisting of records (or connections) between
pairs of entities are ubiquitous in data and network science.
 For applications like peer-to-peer communication, email, various social network 
interactions, and computer network security, it makes sense to organize these 
records into groups based on how and when they are occurring. Weighted line graphs offer a natural 
way to model how records are related in such datasets but 
for large real-world graph topologies the complexity of building and 
utilizing the line graph is prohibitive. We present an
algorithm to cluster the edges of a dynamic graph via the associated line graph without forming it explicitly.

\qquad We outline a novel hierarchical dynamic graph edge clustering approach 
that efficiently breaks massive relational datasets into small sets of edges  
containing events at various timescales. This is in stark contrast to traditional 
graph clustering algorithms that prioritize highly connected (clique-like) community structures.
Our approach relies on constructing a sufficient subgraph of a weighted line graph and 
applying a hierarchical agglomerative clustering.  
This approach is related to scalable techniques from spatial clustering \cite{ester1996dbscan}, 
nonlinear-dimension reduction \cite{mcinnes2020umap}, and topological data analysis \cite{carlsson2021tda} and draws particular inspiration
from HDBSCAN \cite{campello2013hdbscan}.

\qquad We present a parallel algorithm and show that it is able to break 
billion-scale dynamic graphs into small sets that correlate in topology and time.
The entire clustering process for a graph with $O(10 \text{ billion})$ edges 
takes just a few minutes of run time on 256 nodes of a distributed compute environment.   
We argue how the output of the edge clustering is useful for a multitude of data 
visualization and powerful machine learning tasks, both involving the original massive dynamic 
graph data and/or the non-relational metadata.  
Finally, we demonstrate its use on a real-world large-scale directed dynamic graph 
and describe how it can be extended to dynamic hypergraphs and dynamic graphs/hypergraphs with 
unstructured data living on vertices and edges.    
}
{Dynamic Graph Clustering, Line Graphs, Data Analysis, High Performance Computing}


\end{abstract}

\newpage
 
\section{Introduction}

\qquad Dynamic graphs are generally represented in models with vertices that exist throughout time 
and edges that are transient, only existing at a single point in time (or a relatively short 
period of time). Here we assume a collection of temporal edges of the form $e=(i,j,t)$, 
a source/target vertex pair $i,j,$ with a single timestamp of observation, $t$. 
The same source/target pair often participates in multiple edges at different points in time.

Clustering temporal relational datasets is an important and challenging unsupervised 
data analysis task and ideally serves as initial processing for massive 
datasets allowing a broad set of follow-on analysis. An analyst may want to index data such that
 related data can be quickly found or in order to break up data into more manageable chunks. 
 Other examples include aggregating data and using associated metadata of the 
entire cluster to compare each cluster 
with respect to the others and detection of topological or temporal anomalies .

In this work, we will focus on clustering the edges of a dynamic graph calling these clusters 
{\em conversations}, from the natural language 
description of a group of messages between people. It is worth noting that
an edge clustering implies an overlapping node clustering since each node in a network can be a part of
several conversations and each conversation implies a group of members.
The structure of the conversations we find range from fairly sparse, that is path- or 
tree-like, to dense and clique-like, but are generally characterized by relatively 
tight occurrence in time (see Figure~\ref{fig:ex1} for a depiction).

\begin{figure}[h]
\centering
\includegraphics[width=4 in,draft=flase]{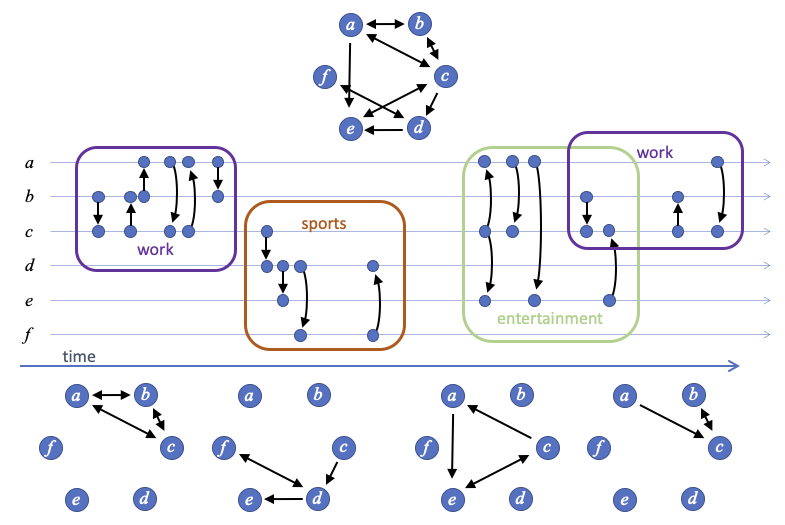}
\caption{ {\bf Example of intermittent community interaction from the perspective of a single vertex.}
\\\hspace{\textwidth}
(Top) Original dynamic multigraph. Each edge potentially represents multiple records occurring at different times.
\\\hspace{\textwidth}
(Middle) Time series of timestamped peer-to-peer interactions.
User $c$ receives a message from a co-worker, $b$, 
which causes $c$ to log on and start interacting with a few different groups of users 
for a small period of time 
before logging off.   First, $c$ reacts to the initial message from $b$, participating in a
work conversation between $a$, $b$, and $c$.   
After a few quick responses, 
experiencing a lull in the conversation, $c$ decides to send a message about a recent sporting event, 
and starts a conversation with $d$, $e$ and $f$.   Next, $c$ starts a conversation 
about a television show with $a$ and $e$.   Lastly, while the entertainment-related conversation is still active, 
$b$ reacts to the last work-related message and the 
work conversation in the beginning of this example is resumed.
\\\hspace{\textwidth}
(Bottom) Flattened snapshots of subgraph topologies associated with each conversational burst.  Note
that some are clique-like involving most possible member pairs, and others are tree-like, but 
all are characterized by relatively quick connected interactions.} 
\label{fig:ex1} 
\end{figure}

\subsection{Classical Community Evolution Detection}

There exist several broad classes of dynamic graph clustering approaches, 
most of which focus on detecting 
classical community structures (clique-like) and their dynamic evolution \cite{Rossetti2018}, 
including the following list. 
 
\begin{itemize} 
\item[(i)] {\bf Windowed Analysis.} One class of techniques employs snapshots, discretizing 
 time globally across the dataset and applying disjoint static graph analysis to the resulting 
graphs that consist of all edges occurring within the time window \cite{graphwindow}.
Some of these approaches factor in neighboring time windows in various ways to smooth and 
track cluster evolution \cite{overlap}. 
\item[(ii)] {\bf Tensor Analysis.} These techniques also discretize time globally and apply 
multi-linear tensor factor analysis to jointly find correlations in timing and topological 
structure \cite{graphtensor}.
\item[(iii)] {\bf Streaming Updates.} These approaches incrementally update the graph or a set of graph features  
by adding (and maybe removing) one or more
edges at a time and recording the changes. \cite{stinger} \cite{graphstream}

\end{itemize}

These approaches have been successful at providing 
insight into many types of dynamic relational datasets,
but they are not without limitations for understanding community structure and evolution.

Classes (i)-(ii) tend to use a global discretizations of time 
(a single time-scale for the entire dataset) in order to provide a coarse-grained view of the data. 
There are several timescales present in many real world datasets where some vertices 
are typically much more active than others. Moreover, there is an underlying trade-off where 
wide time intervals can be used but will tend to 
{\em wash out} coherent events that happen over short time-scales. Conversely, short time 
intervals leave little or even insufficient connectivity.  For most realistic datasets, 
choice of a global partition in time is likely to miss an important relationship 
{\em somewhere} in the data.   To complicate all of this further, many vertices 
often become more or less active at different periods of time, and a global discretization 
may not even be useful for any one vertex.

Streaming algorithms (Class (iii)) have been important for massive datasets that are 
coming in too fast to rebuild models and recompute clusters or features. These algorithms can provide a solution when
 clusters are required as soon as possible for situational 
awareness or a fast response to data trends.  There is a trade-off between 
timeliness and accuracy/complexity, where fast and simple answers are often more important than best answers.
The stability of streaming algorithms can also depend of the order that the edges are processed in streaming update algorithms \cite{comdetection}.

 Another primary challenge is that these techniques tend to focus on
vertex clustering, whereas edge clustering may be a preferable option for some dynamic graph applications.
As an example from network security, sysadmins are already aware of network structure 
and are instead interested in identifying certain types of anomalous events. These events could be 
a user acting in a way (maliciously or otherwise) that will negatively impact the network. 
In this case 
the network maintainers are not interested in a vertex clustering (the bad actor could have an 
overall normal user profile) 
and are instead interested in groups of records (conversations) that contain or imply 
these events.

Line graph analysis offers an attractive solution to all these issues. 
For a given dynamic multigraph $\cG$, the {\em line graph} models 
the relationships between pairs of edges in $\cG$ (see a formal definition in \S\ref{sxn:prelim}). 
Two nodes of a line graph 
are connected only if the corresponding edges in $\cG$ share a node.  
Further, we consider a line graph model that captures the directional and temporal properties of $\cG$
by employing a weighted line graph with edge weights scoring
the strength of connection between records.
We use weights that are a function of the time intervals between the two records to connect 
events that are close in time more strongly than those that happen further apart. 
In \S\ref{generalweights} we discuss a natural generalization to line graph edge weights that are based on other 
 features of the edges and nodes of $\cG$.

Traditionally, line graphs are too expensive to build from large graphs and this problem is only 
exacerbated when explicitly dealing with a multigraph.   For many modern applications
 (i.e. scale-free social media graphs) the 
storage and complexity of clustering line graphs directly is often quadratic in the input.    
Depending on the clustering technique, the computational cost may be even higher. 
Here we 
present a scalable approach that utilizes a sufficient sparse subgraph of the line graph 
and demonstrate its efficacy  and efficiency on dynamic graphs 
ranging from millions to tens of billions of edges.

\subsection{Preliminaries} 
\label{sxn:prelim} 
 
\subsubsection{Notation}

\hfill

Consider a {\em dynamic multigraph model}, where the graph $\cG(\cV, \cE, \cT)$ is a discrete set of 
vertices, $\cV$, and a discrete set of edges, $\cE \subset \cV \times \cV$ , 
that represent relationship events at 
specific times $t \in \cT$ with $\cT$ a continuous interval $[0, T]$.
Each dynamic edge $e_r \in \cE$ is a tuple $e_r = (i_r, j_r, t_r)$ representing a record 
between vertices $i_r, j_r \in \cV$ observed at time $t_r \in \cT$. 
In our discussion, we will assume that $\cG$ is {\em directed}, or that the record
$e_r = (i_r, j_r, t_r)$ represents $i_r$ sending to $j_r$ at time $t_r$. 
Note that $e_r = (i_r, j_r, t_r) \in \cE$ does not imply the existence of a {\em concurrent reciprocal} 
edge $e_r^* = (j_r , i_r, t_r)$, although concurrent reciprocal edges may exist.    
{\em Self loops}, $(i,i,t)$,  are also possible, but their presence and utility is data dependent and 
often require special treatment so we will assume that there are no self loops in $\cG$.

 $\cG$ is a {\em multigraph} which means that there may exist several edges between any two vertices.
In a dynamic multigraph this usually means that the edges exist at different times.
For example, $e_1=(i,j,t_1)$ and $e_2=(i,j,t_2)$, $t_1 < t_2$, represent two different events 
between $i$ and $j$ observed at two different times. 
Furthermore, {\em ties} in timestamps are allowed, that is $e_3=(i,j,t_3)$ and $e_4=(k,l,t_4)$ can 
exist such that $t_3=t_4$. This is common in very large real-world datasets, where ties happen 
due to record volume and discretization of time (e.g. rounding to the nearest second).   
For simplicity, we assume that different dynamic multigraph edges are unique, 
and that no two edges match in their ordered 
vertex pair {\em and} timestamp.

The set of incoming multigraph edges to vertex $i$ is 
$\cE^{(in)}_i := \{ e_r \in \cE \, : \, e_r=(k_r,i,t_r) \}$ and 
the set of outgoing edges is 
$\cE^{(out)}_i := \{ e_r \in \cE \, : \, e_r=(i,j_r,t_r) \}$. 
The {\em vertex in-degree} of $i \in \cV$, $d^{(in)}_i := |\cE^{(in)}_i|$, is
the number of multigraph edges with $i$ as a target.
Similarly, the  {\em vertex out-degree} is $d^{(out)}_i := |\cE^{(out)}_i|$. 
The set of all edges incident to vertex $i \in \cV$ is $\cE_i := \cE^{(in)}_i \cup \cE^{(out)}_i$
and the {\em vertex degree} of $i$ is $d_i := |\cE_i| = d^{(in)}_i + d^{(out)}_i$, 
the total number of multigraph edges incident to $i$. 
These quantities have maximal values across $\cV$,  
$d_{max} := max_{i \in \cV} \, d_i$.
The {\em vertex neighborhood} of $i \in \cV$, written $\cN_i$, is the set of all vertices with
one or more edges incident to $i$.  Due to the possibility of multiple edges involving 
the same pair of vertices, $|\cN_i| \leq d_i$ in general.

\begin{table}
\rowcolors{2}{gray!25}{white}
\centering
\begin{tabular}{|c|c|c|}
\hline
Symbol & Definition & Section \\
\hline
$\cG(\cV, \cE, \cT)$ & Dynamic multigraph, vertices, edges, time &  \\
$e = (i,j,t)$  & Edge from node $i$ to node $j$ at time $t$ & \\
$\cE_i$ & Edge incidence neighborhood & \S\ref{sxn:prelim} \\
$d_i, d_{max}$ & Vertex degree, maximum degree & \\
$\cN_i$ & Vertex Neighborhood & \\

\hline
$\cM_\cG = MST(\cG)$ & Minimum Spanning Tree of $\cG$ & \\
$\cL_\cG(\cE, \cF, \cW)$ & Increment-weighted Line Graph & Def \ref{def:lg} \\

$\cL_\cG(i)$ & Node Local Increment Weighted Line Graph & Def \ref{def:nliwlg} \\

\hline
$\cL^*_\cG$ & Line Graph Skeleton of $\cG$ & Alg \ref{alg:skeleton}\\

$\cM_\cG(i)$ & Node Local MST($\cL_\cG$) & \\

\hline
$\mbox{comps}(\cH, \omega)$ & $\omega$-weighted connected components of $\cH$ & \S\ref{sec:distributed} \\

\hline
\end{tabular}
\caption{Notation}
\label{tab:notation}
\end{table}

\newpage

\subsubsection{Increment-Weighted Line Graph}
\label{linegraphdef}

\begin{figure}[h]
  \centering
  \includegraphics[width=4.75 in,draft=false]{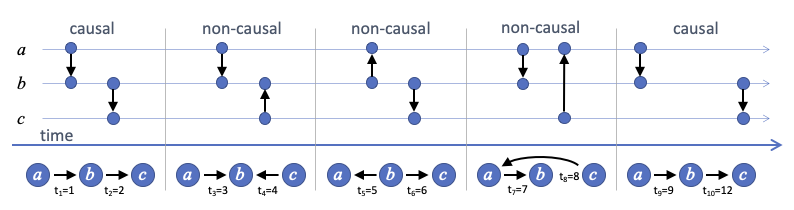}
  \caption{ {\bf Causality for pairs of incident timestamped edges. } 
  Increment-weighted line graph edges represent causal flow of information according to our 
  model, as described formally in Def.~\ref{def:lg}. 
  (Left) Tip-to-tail directed edges that are ordered correctly in time represent  
  multi-hop causal information flow. (Center-Left) In contrast, pairs of directed edges incident 
  to a sink vertex are not causal and (Center) similarly for a source vertex. 
  (Center-Right) Tip-to-tail in the wrong order is also not causal.   
  (Right) Larger gaps in time between causal edge pairs are represented by larger weights.
  } 
  \label{fig:causal} 
\end{figure} 

 For an undirected graph $\cG$, the line graph is formed by creating a vertex
 for each edge in $\cG$ and connecting those line graph vertices if the corresponding edges in $\cG$ share a node. 
 Note that this definition extends trivially to a multigraph. 
 Here, we define a directed variant of the line graph that heuristically follows the potential flow of
 information, where two line graph nodes are 
 connected only when the corresponding edges in $\cG$ share a forward relation, or are {\em tip-to-tail}  
 (see Figure~\ref{fig:causal}). 
 This can be viewed as enforcing a causal relation between nodes of the line graph, i.e. two nodes of 
 the line graph are only related if the corresponding records in $\cG$ were aligned such that one interaction 
 could have directly caused the other (based on evidence present within the dataset alone).

Finally, each edge of the line graph is weighted based on how {\em close} the associated records are.
Let the distance between two records be defined as the time interval between when they occurred (see a more 
general idea in \S \ref{generalweights}). 
This implies that two records (nodes of of the line graph) that happen closer in 
time are more strongly connected than 
pairs with a larger gap between them. We will, for simplicity, just use the difference in time between records
as a weight meaning that a smaller weight implies a stronger connection 
(see Figure~\ref{fig:causal}). 

\begin{definition} {\sc (Increment-Weighted Line Graph)}
\label{def:lg}
Let $\cL_\cG(\cE, \cF, \cW)$ be 
the {\em Increment-Weighted Line Graph} of $\cG(\cV, \cE, \cT)$, written $\cL_\cG$ for short. 
Here, $\cE$  plays the role of the vertex set of $\cL_\cG$ and  $\cF \subset \cE \times \cE$, is the set
of relationships between nodes of $\cL_\cG$ (edges of $\cG$). The set $\cW$ contains increment weights 
$w_{rs} \in \mathbb{R}^+$.   
Specifically, given  $e_r, e_s \in \cE$ with $e_r = (i_r, j_r, t_r)$ and $e_s = (i_s, j_s, t_s)$, the line graph edge $f_r = (e_r, e_s, w_{rs}) \in \cF$ 
if and only if the following properties are satisfied.

\begin{itemize}
\item[{\bf (P1)}] {\bf Tip-to-tail Connectivity.} $j_r = i_s$. 
\item[{\bf (P2)}] {\bf Temporal Causality.}                $t_r \leq t_s$.            
\end{itemize}

And the weights of $\cL_\cG$ are defined as:

\begin{itemize}
\item[{\bf (P3)}] {\bf Increment Weights.}        $w_{rs}=t_s - t_r$.            
\end{itemize}
\end{definition}

{\sc Note:} In a slightly different definition that uses {\bf (P2')} $t_r < t_s$, with strict inequality, 
the resulting increment-weighted line graph would be a {\em directed acyclic graph} (DAG), 
which fits our understanding that it is modeling the flow of information in a dynamic graph $\cG$.    
The absence of cycles would be equivalent to information not flowing backwards in time.  
However, many large realistic datasets may have relatively quick-acting information flows with respect
to their time resolution (say timestamp fidelity is at the hour level and we would like to capture information 
propagations that happen at the minute level). We therefore use the non-strict inequality, 
{\bf (P2)} $t_r \leq t_s$, which results in a mostly DAG-like line graph $\cL_\cG$ that 
has some strongly connected components where connected ties in timestamps occur in the data.

Typically $\cL_\cG$ is a much-reduced subgraph of the complete line graph 
(the one formed without the causality and tip-to-tail restrictions), 
but in general $\cL_\cG$ will have too many line graph edges to build from a large 
graph $\cG$, due to the quadratic scaling induced by the existence of high degree nodes.

\subsubsection{Hierarchical Agglomerative Clustering}

\hfill

We turn to a promising method of clustering that also contains a solution for the problem
 of generating and storing the increment weighted line graph $\cL_\cG$.  
Drawing heavy inspiration from HDBSCAN, 
consider clustering a graph using a hierarchical agglomerative algorithm. 
There are many benefits to this form of clustering such as needing minimal a priori knowledge of your data, 
allowing clusters to exist at many different scales w.r.t. the graph weighting, 
and allowing for a set of nodes that doesn't belong to any cluster. 
These are all excellent reasons to use this method of clustering but the most important
 benefit is that the computation only requires a minimum spanning tree
 of the input graph.

 The algorithm described in  \cite{campello2013hdbscan} is intended to run on data where there exists an all-to-all distance 
 between points. This set of distances generates a dense graph which is generally too large to work with.
  The first step in the HDBSCAN algorithm when dealing with
  this dense graph is to compute a minimum spanning tree (MST). This serves a dual purpose in that the MST
  is much smaller and easier to work with, and the clustering algorithm only needs the information 
  found in the MST. In other words the portion of the graph that is not contained in the MST is redundant.

 Our approach begins with a weighted 
  graph that is generally much sparser than the fully connected weighted clique, 
  and proceeds with the latter parts of the HBDSCAN algorithm. HDBSCAN requires one
  input parameter, $M$, that defines what the smallest allowed cluster size is. This {\em minimum cluster size}
  is generally chosen based on the dataset and the desired size for the clusters.
  Leaving additional details of the clustering algorithm to \cite{campello2013hdbscan}, 
  given a weighted graph $\cG$, the hierarchical agglomerative clustering algorithm we 
  will be using is outlined in Algorithm~\ref {alg:hier_cluster}.

\begin{algorithm} 
\caption{{\bf Hierarchical Agglomerative Clustering} }
\label{alg:hier_cluster} 
\begin{algorithmic} 
 
\State {\bf Input:} Weighted Graph $\cH(\cV,\cE,\cW)$, minimum cluster size $M$.

\State {\bf Output:} Set of node clusters. 
 
\State 

\State 1. Compute minimum spanning tree of $\cH$: $\text{MST}(\cH)$.

\State 2. Sort edges of $\text{MST}(\cH)$ based on weights.

\State 3. Build the cluster hierarchy (dendrogram) using a union-find data structure.

\State 4. Condense the cluster hierarchy using $M$ as a threshold for which clusters are allowed.

\State 5. Extract clusters based on {\em volume} of condensed clusters in  {\em dendrogram space}.

\end{algorithmic} 
\end{algorithm}

The graph we cluster is the increment-weighted line graph $\cL_\cG$, but since $\cL_\cG$
is directed, there are different notions of connectivity to consider.
We use weak connectivity since we are interested in
patterns such as the one-way flow of 
information throughout a directed graph
which only require one-way reachability between nodes. 
We achieve this type of clustering by replacing the directed edges of $\cL_\cG$ with undirected edges and 
applying Algorithm~\ref {alg:hier_cluster} to the resulting undirected version of $\cL_\cG$.
The clusters of the undirected version of $\cL_\cG$ are then interpreted as weakly connected clusters
of $\cL_\cG $.

For large enough datasets, it is computationally expensive to apply 
Algorithm~\ref {alg:hier_cluster} to the $\cL_\cG$ since it is unreasonable 
to even just build and store $\cL_\cG$.   Fortunately, the first step of the clustering algorithm
 implies that we do not need all the edges of $\cL_\cG$. In section \ref{sec:iwlg} 
 we will present a suitable subgraph of $\cL_\cG$ called a {\em Line Graph Skeleton} that is bounded in size
 and can be computed directly from the original dynamic graph $\cG$. We also show that the Line Graph Skeleton can
 be constructed in a highly parallelizable way and shares a connectivity structure with $\cL_\cG$.
 In Section \ref{sec:distributed} we cover modifications
 to Algorithm~\ref {alg:hier_cluster} for datasets that necessitate distributed scale hardware for storage and computation.
 Section \ref{sxn:experiment} contains two examples of clustering the edges of dynamic graphs with section \ref{sxn:small_experiment} focusing on a smaller dataset and presenting many of the unique properties of our 
 clustering algorithm. Section \ref{sxn:large_experiment}, on the other hand,
  demonstrates the scalable performance on a much larger graph. We finish with a discussion of several 
  extensions and generalization to the work of this paper in Section \ref{sxn:extensions}.

\section{Efficient Increment-Weighted Line Graph Clustering} 
 
 \label{sec:iwlg}
 
Here, we describe an efficient approach to the line graph clustering described in the previous section.
Important theoretical properties related to this approach are proven in the 
Appendix~\S \ref {sxn:theoryappendix}.

A large, real-world graph often has a heavy-tailed vertex degree distribution. 
The cost of explicitly building and computing graph analytics for an associated line graph is 
well-known to be prohibitively expensive, as the number of edges in the line graph is:

$$ 
|\cF| = \sum_{i \in \cV} {d_i \choose 2}. 
$$ 

In many real-world applications $d_{max}$ is order $|\cV|$ and thus $|\cF|$ is order 
$|\cV|^2$, and computation is not practical for applications involving even as few as millions of vertices and edges.   
Even though $\cL_\cG$ is a subgraph of the full line graph, it still suffers from the 
same quadratic scaling and will also be difficult or impossible to generate for general dynamic graphs.
We take a vertex-local approach to implicitly represent the connectivity of $\cL_\cG$ without forming it explicitly.

Heuristically, an edge in $\cL_\cG$ exists between two nodes of $\cL_\cG$ if the corresponding edges
of $\cG$ share a node in the appropriate way.
This means that every edge of the line graph comes from the neighborhood of one single node of the original 
graph $\cG$. In fact, we are able to build $\cL_\cG$ by having each node of $\cG$ build the portion of the line graph that it
is responsible for and then performing a union over all of the {\em node local line graph} pieces. More formally:

\begin{definition} {\sc Node Local Formation of Increment-Weighted Line Graph}
\label{def:nliwlg}
\\\hspace{\textwidth}
For each node $i \in \cG$ define the {\em Node Local Increment Weighted Line Graph} or $\cL_\cG(i)$ 
to be the subgraph of $\cL_\cG$ that is generated from applying Definition \ref {def:lg} to $\cE_i$. 
The union of $\cL_\cG(i)$ over all the nodes in $\cG$ will result in the Increment Weighted Line Graph:
$$\cL_\cG = \cup_{i\in \cV} \cL_\cG(i).$$
\end{definition}

This means that $\cL_\cG$ can be generated in a highly parallelizable way with each 
node of $\cG$ contributing $\cL_\cG(i)$ independently of the other nodes. 
This still does not address the issue of $\cL_\cG$ being too large to even hold in memory.
Fortunately, Algorithm~\ref {alg:hier_cluster} implies that the entire graph is never used
for more than generating a minimum spanning tree. 

What we actually need is a minimum spanning tree of $\cL_\cG$, denoted
$\cM_\cG = $ MST$(\cL_\cG)$, or a much sparser subgraph of $\cL_\cG$ that contains an MST. 
Instead of creating $\cL_\cG$ 
and then applying an algorithm to find $\cM_\cG$, we present a method to compute
a bounded subgraph of $\cL_\cG$ with the same connectivity of $\cM_\cG$.
Taking inspiration from Def.~\ref{def:nliwlg}, we will compute this subgraph of $\cL_\cG$  
in a similar node-local way.  For each node of the original graph $\cG$, compute a local minimum spanning tree instead 
of having each node compute its full contribution to $\cL_\cG$. The task for each node
is detailed in Algorithm~\ref {alg:nlmst}. Putting this all together yields 
the {\em Implicit Increment Weighted Line Graph} or {\em Line Graph Skeleton} (denoted $\cL^*_\cG$) of $\cG$ detailed in Algorithm~\ref {alg:skeleton}.

\begin{algorithm} 
\caption{Generate Node Local MST($\cL_\cG$) $:= \cM_\cG(i)$} 
\label{alg:nlmst}
\begin{algorithmic} 
\State Input:  Dynamic graph $\cG$ and a node $i \in \cG$ 
\State Output:  Edgelist of $\cM_\cG(i)$
\State 

\State $\cE_i :=$ Edge incidence neighborhood of $i$
\State Sort $\cE_i$ by time.

\State $W=$ [ ] \Comment{empty array} 
\State  output $=$ [ ] \Comment{empty array} 
 \For{ edge $e_j \in \cE_i$ } 
 \If{$e_j \in \cE^{(in)}_i$} 
    \State $W.\text{append}(e_j)$ 
\Else 
   \For{$e_k  \in W$} 
        \State $\text{output.append}(\, (e_k,e_j)\, )$ 
   \EndFor 
     \State $W = [W[-1]] $ \Comment{The in-edges in W have been included in the MST; delete all but the last} 
\EndIf 
\EndFor 

\State \Return output
\end{algorithmic} 
\end{algorithm} 

A few high-level constructs are employed in the listing of Algorithm~\ref{alg:nlmst} that are 
likely familiar to python developers.  Given a {\em mutable array} or {\em list} called \texttt{mylist}, 
the operation \texttt{mylist.append(item)} adds ``item" as the last element of the list.  Second, \texttt{mylist[-1]}
returns the last value in \texttt{mylist}. Lastly, \texttt{mylist=[]} and \texttt{mylist=[item]} creates an empty list and 
a single element list with ``item" in it, respectively. See Figure \ref{fig:skeleton} and Appendix \ref{sxn:theoryappendix} for 
an example illustrating Alg.~\ref{alg:nlmst} along with a more in-depth mathematical discussion.

\begin{algorithm} 
\caption{Generate {\em Line Graph Skeleton} of $\cG:= \cL^*_\cG$}
\label{alg:skeleton}
\begin{algorithmic} 
\State Input:  Dynamic Multigraph $\cG$

\State Generate adjacency list representation for $\cG$

 \For{ node $i \in \cV$ } 
    \State Compute Node Local MST($\cL_\cG$): $\cM_\cG(i)$ \Comment{Alg.~\ref {alg:nlmst}}

\EndFor

\State Compute Implicit Increment Weighted Line Graph $\cL^*_\cG$:
  
$\cL^*_\cG := \cup_{i\in \cV} \cM_\cG(i)$

\end{algorithmic} 
\end{algorithm}

Algorithm~\ref {alg:skeleton} finally gives us what we need in order to perform an edge clustering on 
the original dynamic graph $\cG$. This {\em Line Graph Skeleton} of $\cG$, $\cL^*_\cG$,
has two properties that make 
it an ideal line graph substitute for the type of clustering we are planning to achieve. First, the number of 
edges of  $\cL^*_\cG$ is bounded by $2 |\cE|$ (twice the original number of edges in $\cG$):
 
$$
|\text{edges}(\cL^*_\cG)| = 
\sum_i |\text{edges}(\cM_\cG(i))| \leq 
\sum_i (|\cE_i| -1)= 
2|\cE| - |\cV|
$$

The inequality comes from the possible reduction in connections of $\cL_\cG$ due to 
the causality and tip-to-tail constraints. The final factor of two comes from the fact that
every edge of $\cG$ will be in the neighborhood of two nodes. The total number of 
nodes in $\cG$,  $|\cV|$, appears because the MST for a graph with $N$ nodes 
will have at most $N-1$ edges.

The second important property of $\cL^*_\cG$ is that it shares a connectivity structure with 
the minimum spanning tree of $\cL_\cG$. In Theorem~\ref{thm:lasttheorem} (in the Appendix) 
we show that the line graph 
skeleton $\cL^*_\cG$ has equivalent 
weight-filtered connected components to $\cL_\cG$ or:
$$ 
\mbox{comps}( \cL^*_\cG, {\Delta t}) \equiv  \mbox{comps}( \cL_\cG, {\Delta t}). 
$$ 
 meaning that we are free to cluster the bounded $\cL^*_\cG$ in place of the $\cL_\cG$.

We summarize the dynamic graph edge clustering algorithm of Section \ref{sec:iwlg} in Algorithm \ref{alg:edgecluster}.

\begin{algorithm} 
\caption{{\bf Dynamic Graph Edge Clustering} }
\label{alg:edgecluster} 
\begin{algorithmic} 
 
\State {\bf Input:} Weighted Dynamic Multigraph $\cG(\cV,\cE,\cW)$, minimum cluster size $M$.

\State {\bf Output:} Set of edge clusters. 
 
\State 

\State 1. Compute Line Graph Skeleton of $\cG : \cL^*_\cG$ using Alg. \ref{alg:skeleton}.

\State 2. Run Alg. \ref{alg:hier_cluster} on the undirected version of $\cL^*_\cG$

\State 3. The node clustering of  $\cL^*_\cG$ is the desired edge clustering of $\cG$

\end{algorithmic} 
\end{algorithm}

\section{Distributed memory scalable algorithm}

\label{sec:distributed}

Many graph datasets large enough to warrant distributed scale compute 
environments for both storage and computation. 
In this section we present a parallel adaptation of Algorithm~\ref {alg:edgecluster} for such graphs.
The first step is to generate $\cL^*_\cG$ given a graph $\cG$.
Algorithm~\ref {alg:skeleton} is already well-structured for distributed computing. 
Each node of $\cG$ only needs its own set of edges to compute its contribution
to $\cL^*_\cG$ . In addition, each node's contribution takes up less space than the size 
of the neighborhood of that node, i.e. for any node $i$ in $\cG$: 
the number of edges of $\text{MST}(\cL_\cG(i)) < |\cE(i)|$. 
This means that the only communication needed between 
compute nodes happens at the initial step of building the adjacency list and at the end when 
performing the union of  $\cM_\cG(i)$ and subsequent sorting of the edges of the MST 
in step 2 of Algorithm~\ref {alg:hier_cluster}.
 
For large problems, iterating through all the edges of $\cL^*_\cG$ in order to produce the 
dendrogram in step 3 of Algorithm~\ref {alg:hier_cluster} is an expensive serial operation.
To deal with this issue we introduce an approximation to the dendrogram in the form of a series 
of connected components with increasing weight thresholds. 
Any level $w_i$ of the true dendrogram can be interpreted as the resulting connected components 
of $\cL^*_\cG$ if you remove edges of $\cL^*_\cG$ with weight greater than $w_i$. 
For the true dendrogram this set of components is then incrementally updated as each new edge of 
$\cL^*_\cG$ is added in increasing weight order. In order to discretize the dendrogram, 
we add many edges at once in batches to 
reduce the number of levels of the dendrogram from $O(|\cE|)$ to $N$, where $N$ is a constant chosen so that we 
can compute $N$ connected component runs in a reasonable time. This has the effect of creating 
a coarse representation of the single linkage dendrogram and is just an approximation of the
true dendrogram, but it is still representative of the underlying structure as long as we are free to choose
a reasonable set of weight thresholds.

\begin{figure}[h]
\centering
\begin{subfigure}{.5\textwidth}
  \centering
  \includegraphics[width=.9\linewidth]{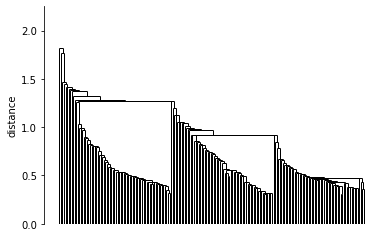}
  \subcaption{Dendrogram formed by adding each edge in series.}
  \label{fig:cliquegraph}
\end{subfigure}%
\begin{subfigure}{.5\textwidth}
  \centering
  \includegraphics[width=.9\linewidth]{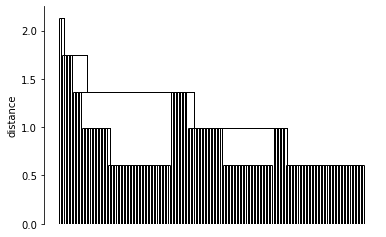}
  \subcaption{Coarse dendrogram formed using 5 levels.} 
  \label{fig:stargraph}
\end{subfigure}
\caption{The X axis corresponds to a sorted ordering of the nodes based on the distance threshold where the nodes are added to the MST. This is represented in the dendrogram via a horizontal line at that height and sorting the nodes such that there are no line crossings.}
\label{fig:dendro}
\end{figure}

Given a weighted graph $\cH(\cV,\cE,\cW)$ and a set of thresholds 
$[ \omega_1, ... \omega_q ]$ we compute the {\em discrete dendrogram} as follows.
For each threshold, $\omega_i$, only allow edges with weight less than $\omega_i$
and compute the connected components denoted: $\mbox{comps}( \cH, \omega_i)$. 
For any two thresholds such that $\omega_2 > \omega_1$, we see that the connected components 
are {\em nested} in the following sense. For any node $v \in \cV$, let $\cH_1$ be the component in 
$\mbox{comps}( \cH, \omega_1)$ containing $v$ and $\cH_2$ be the component in 
$\mbox{comps}( \cH, \omega_2)$ also containing $v$. The nested property implies 
$\cH_1 \subset \cH_2$. Therefore, we can represent the relationship between components 
at two different thresholds as an edge $(\cH_1 \subset \cH_2)$ 
within a dendrogram $\cD$, or a tree that stores the 
parent-child containment relationships between nested partitions. 
This process is summarized in Algorithm \ref{alg:build_dendro}.

\begin{algorithm} 
\caption{Build Discrete Dendrogram} 
\label{alg:build_dendro} 
\begin{algorithmic} 
 
\State {\bf Input:} Weighted Multigraph $\cH(\cV,\cE,\cW)$, 
ordered set of increasing increments $[ \omega_1, ... \omega_q ]$. 
 
\State {\bf Output:} Dendrogram $\cD$. 
 
\State

\For{ $\Delta \omega_i \in [ \omega_1, ... \omega_q ]$} 
 
\State Compute $\mbox{comps}( \cH, \omega_i)$. 
 
\For{ $\cH_k \in \mbox{comps}( \cH, \omega_i)$} 
 
\State Add component $\cH_k$ as parent vertex in $\cD$ 
\State Add to $\cD$ edges from parent vertex $\cH_k$ to 
all constituent components (child vertices)

\EndFor 
\EndFor 
 
\State \Return Dendrogram $\cD$. 
 
\end{algorithmic} 
\end{algorithm}

The discrete dendrogram does represent a loss in resolution,
but this is an acceptable tradeoff when the more detailed version is not 
computable. In Figure \ref{fig:dendro} the dendrogram is discretized using 5 equally 
spaced bins, and the overall structure is maintained. Because the dendrogram is 
produced using the connectivity of the minimum spanning tree of $\cL_\cG$,
there are going to be many more connections taking place at lower weight 
thresholds. This implies that for most datasets it would be better to use
non-uniform bin spacing with more smaller bins at lower weight 
thresholds and fewer, larger bins for higher thresholds.

This brings us to Step 4 of Algorithm \ref{alg:hier_cluster}. It can be performed as 
described in \cite{campello2013hdbscan} on the discrete dendrogram. Normally the runtime
 of HDBSCAN 
in this step depends on the minimum cluster size parameter. This makes sense
 as a lower minimum cluster size means that more potential clusters in 
the dendrogram will need to be considered.
 For the discrete dendrogram each bin or level is handled simultaneously, i.e.
there will only be a small number of steps required (equal to the number of
bins used to generate the dendrogram) regardless of what the minimum cluster size is.
This means that when using the discrete dendrogram the limiting parameter for this
step is now the number of levels used and removes the scaling dependence on minimum 
cluster size. The tradeoff is that a coarser discretization of the dendrogram
might not have the resolution necessary to find small clusters.

\section{Experiments} 
 \label{sxn:experiment}

Analysis involving one or more modern social network platforms can easily reach internet or 
global scale datasets, involving the user accounts of millions of humans and accounts driven by 
computer programs ({\em bot activity}).   
Modeling the peer-to-peer interactions on these platforms as a dynamic multigraph facilitates trend analyses, 
link prediction, anomaly detection,
behavior analyses, and other topology-related analytics, and clustering of the records is an important line of 
approach to aide in such efforts.

Imagine an archetype platform for this discussion.
The human users interact on the platform at different rates, some use it sporadically, 
some a few times a week, some daily.   {\em Super users} are engaging on the platform almost 
the entire time they are awake.   
Moreover, a single user might interact with some topics heavily and others intermittently. 
If the platform is globally popular the users can live all over the globe and different timezones 
will dictate when they tend to be more active.   
Topical communities involving many users exist within the 
platform, but they are often woven through time in complicated ways, as some communities have user 
constituents from many timezones that participate in the community at various timescales.

There are many reasons for bot accounts ranging from functional such as giving instructions and tips, to 
users who invoke them to exploitative bots masquerading as human users.
The bots may have super-human participation rates or be programmed to act more like an actual user.
For analytics aimed at human activity, 
like serving advertisements, it may be ideal to filter out computer generated records while 
other analytics may be focused on the bot activity itself or even just identifying bot accounts.

In this section we demonstrate the line graph-based clustering on two such peer-to-peer temporal datasets.
In \S\ref{sxn:small_experiment} we apply the serial techniques to a relatively small set of openly available 
anonymized email sender-receiver records with timestamps.   Here we show many details of the clustering 
output, including cluster sizes and durations, and many graph properties of the resulting clusters. 
In \S\ref{sxn:large_experiment} we apply the distributed techniques to a massive set of social media 
author-to-author comments.   
In this case, we filter out records associated with the largest known bot account, 
and measure the strong 
scaling of the most expensive algorithmic phases, namely implicit line graph construction and dendrogram
computation.

\subsection{Cluster Analysis}
\label{sxn:small_experiment}

The clustering that results from applying Algorithm \ref{alg:edgecluster} to a dynamic graph has many desirable 
features when it comes to understanding the underlying structure of the data. This is, in part, due to the 
relatively unique clusters formed when compared to the generally clique-like clusters produced by
other graph clustering algorithms. 
 Each edge is a part of only one cluster but because we are dealing with a multigraph, 
any given node can be a participant in any number of clusters. 
Additionally, each cluster is
a group of edges which can be interpreted as a connected subgraph.
In this section, we will present a high level   
overview of the clusters found by applying Algorithm \ref{alg:edgecluster} to the {\em Email Network} 
dataset presented in \cite{emaildata}.

\subsubsection{Data Overview}

\hfill

The graph is derived from an email network belonging to a large European
 research institution. In particular, we will be using the instance of the data hosted on 
 the SNAP database \cite{snapnets}. The data spans 18 months (525 days) and consists only of the 
 {\em core} of the network where every edge is both to and from one of the members of the institute.
 The nodes are anonymized representations of the people and each email corresponds 
 to a number of directed edges equal to the number of recipients of that email, each with the same timestamp. 
  The resulting graph contains 986 nodes and 329,910 edges.

\subsubsection{Clustering}

\hfill
 
The {\em line graph skeleton} that is formed during the clustering process 
is indeed larger than the input graph with 329,910 nodes (one for each edge of the original directed multigraph)
and 646,653 weighted edges, falling just shy of the bound presented in \S \ref{sec:iwlg}. 
 The clustering was performed with a minimum cluster size of 5, meaning that any record 
needs to be connected in a group of at least 5 (as prescribed by Def. \ref{def:lg}) before
they are considered a cluster candidate. This threshold results in a clustering that 
designates 152,074 records, 
or about $46\%$ of the original edges, as not belonging to any cluster. 
This designation can be a powerful filtering tool and is 
controlled by the minimum cluster size where a larger value results in more edges left unclustered.

The remaining 177,836 edges are contained in 15,986 clusters ranging in size from 
5 to 175 with the distribution shown in Figure~\ref{fig:size_distro}. The most striking
observation here is that there is a large number of small clusters with no semblance of a 
giant component. Looking at the number of clusters per day in Figure~\ref{fig:cluster_per_day} 
 reveals that there 
are only about 50 clusters active on any given weekday and around 5 per day for the weekend.
Each node of the original graph can participate as a member of many different clusters 
as seen in Figure~\ref{fig:node2cluster}. This is due to the nature of an edge clustering;
each node has potentially many connecting edges and each of those edges could be placed
in different clusters. 

Figure \ref{fig:size-duration} compares the number of emails involved in any given 
cluster to the total duration of the cluster, i.e. the time between the first and last emails. 
Nearly all of the clusters have a duration less than 8 hours with a preference 
for shorter clusters around an hour in duration. This preference is an artifact of the 
distance function that puts more weight on connections between records that are happening closer to 
each other in time. Even with this weighting, there are clusters that exist 
across several days, as seen in Table~\ref{tab:longcluster}.
This all suggests that the clustering was able to capture the relatively tiny scale of several 
hours consistently over the 18 month timeframe.

\begin{figure}[h]
  \centering
  \includegraphics[width=3.0 in]{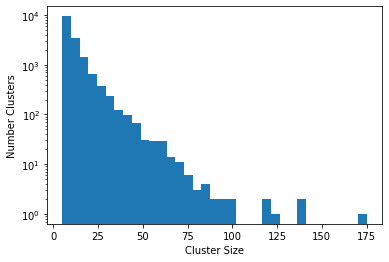}
  \caption{Log-scale histogram of cluster sizes. More than half of the clusters contain fewer than 10 records and only a few contain more than 100 records.} 
  \label{fig:size_distro} 
\end{figure}

\begin{figure}[h]
  \centering
  \includegraphics[width=5.7 in]{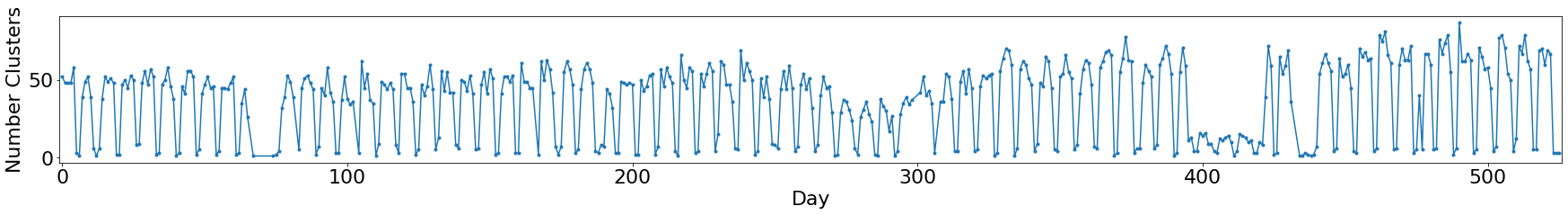}
  \caption{Number of clusters with an email exchange on the given day. There is an obvious pattern of weekdays vs weekends along with lulls in activity that likely correspond to holidays. There is also a general increase in cluster count over time.} 
  \label{fig:cluster_per_day} 
\end{figure}

\begin{figure}[h]
  \centering
  \includegraphics[width=3 in]{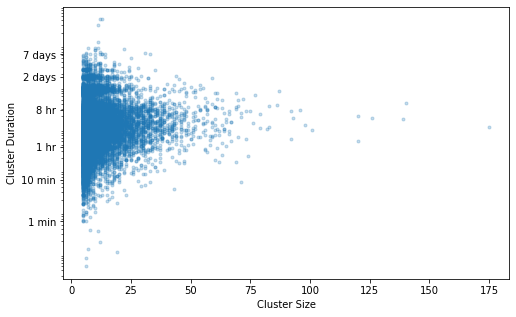}
  \caption{Cluster size (number of emails) compared to the duration of the cluster (time between first and last record). The preference for tighter connections in time and relatively small min cluster size of 5 means that the bulk of the clusters fall under 8 hours in duration. There are also visible artifacts of slightly denser regions at y-axis heights corresponding to durations of 2 and 3 days as well. } 
  \label{fig:size-duration} 
\end{figure}

\begin{figure}[h]
  \centering
  \includegraphics[width=3 in]{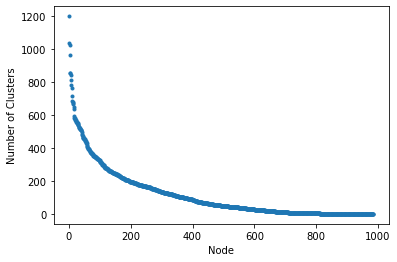}
  \caption{A plot of how many clusters each node of the original graph is a member of. The nodes here are sorted based on the total number of clusters (conversations) they participate in. Almost half the nodes are seen in 50 or fewer clusters and the most prolific node is in 1,198 clusters.} 
  \label{fig:node2cluster} 
\end{figure}

\newpage
 
\subsubsection{Example clusters}

\hfill
 
The purpose of this section is to explore a set of clusters chosen to highlight some 
of the unique aspects of the clusters. It is only meant to provide a brief glimpse 
into the clustering and not an in-depth or exhaustive analysis. 
Starting with the largest cluster shown in Figure~\ref{fig:biggraph}, we already see some noteworthy features. 
The subgraph induced is not very clique-like and most nodes are only connected to one or two
other nodes within this cluster of edges. This branching behavior can be seen in an extreme form in 
Table~\ref{tab:chain_graph} and Figure~\ref{fig:diametergraph} which both have chains
of emails happening in succession. Table~\ref{tab:chain_graph} is a simple example of a
relay where each email could be passing forward information from the previous. 
Figure~\ref{fig:diametergraph} is slightly more involved with 2 branches emanating from a source.

Some clusters do present as more clique-like, although it does seem less common in this particular 
dataset where most records are just one sender to one receiver. The two examples of
Figure~\ref{fig:cliques} feature emails that are sent to multiple recipients and are likely met
with one or more ``reply all" emails that create a tightly connected cluster both topologically and temporally.

\begin{figure}[h]
  \centering
  \includegraphics[width=3 in]{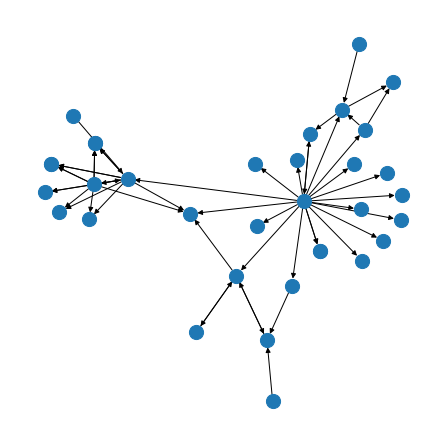}
  \caption{The graph induced by the largest cluster of edges. There are 175 emails exchanged over 3 hours of time between 34 total participants.} 
  \label{fig:biggraph} 
\end{figure}

\begin{table}
\rowcolors{2}{gray!25}{white}
  \centering
  \begin{tabular}{|c|c|c|}
  \hline
  Sender & Receiver &  Time\\
  \hline
  987 & 433 & 08:02:52 \\
  433 & 32 & 08:31:20 \\
  32 & 977 & 09:04:17 \\
  977 &418 & 09:09:08 \\
  418 & 996 & 09:19:59 \\
  \hline
  \end{tabular}
  \caption{A simple cluster that represents a chain of emails over about an hour and a half.}
  \label{tab:chain_graph}
\end{table}

\begin{figure}[h]
\centering
\begin{subfigure}{.5\textwidth}
  \centering
  \includegraphics[width=.8\linewidth]{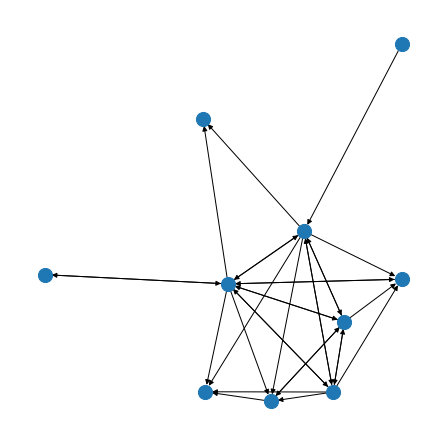}
  \label{fig:cliquegraph}
\end{subfigure}%
\begin{subfigure}{.5\textwidth}
  \centering
  \includegraphics[width=.8\linewidth]{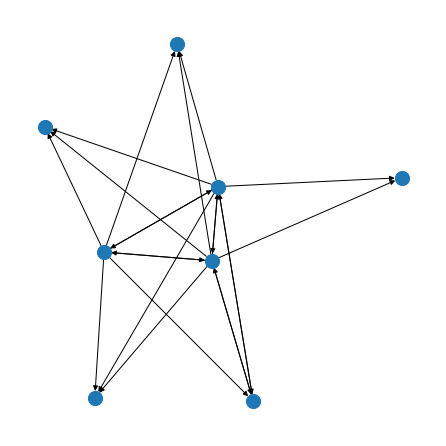}
  \label{fig:stargraph}
\end{subfigure}
\caption{Two highly connected clusters. Both contain multiple ``reply all" style emails over the course of a day. }
\label{fig:cliques}
\end{figure}

\begin{figure}[h]
  \centering
  \includegraphics[width=3 in]{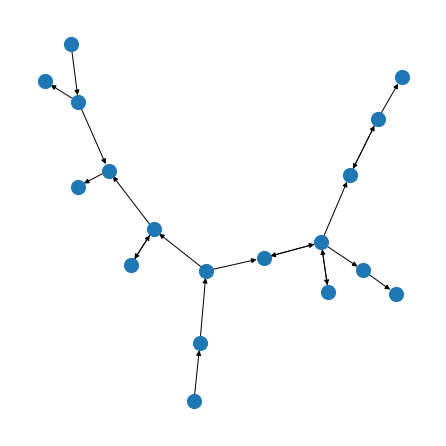}
  \caption{A conversation between 18 people lasting 7 hours. There are 21 emails sent starting from the bottommost node and then branching up to the left and right over the day.} 
  \label{fig:diametergraph} 
\end{figure}

 \begin{table}
\rowcolors{2}{gray!25}{white}
  \centering
  \begin{tabular}{|c|c|c|c|}
  \hline
  Sender & Receiver & Day &  Time\\
  \hline
  317 & 441 & 479 &  14:21:00 \\
  441 & 317 & 490 & 08:27:14 \\
  441 & 304 & 490 & 09:26:06 \\
  441 & 304 & 491 & 12:31:39 \\
  304 & 317 & 491 & 13:14:30 \\
  317 & 304 & 492 & 07:50:42 \\
  304 & 317 & 492 & 13:00:39 \\
  317 & 304 & 492 & 14:41:14 \\
  304 & 317 & 493 & 06:37:41 \\
  304 & 317 & 493 & 12:17:08 \\
  317 & 304 & 493 & 12:20:24 \\
  \hline
  \end{tabular}
  \caption{A cluster that spans several days. In the email record landscape this cluster with only 2-3 emails per day for four days was still able to stand out as significant. The inclusion of the first record, even though it happened 10 days prior to the second, is most likely because it did not connect to any other cluster and its inclusion here still satisfies the {\em causality rules}.}
  \label{tab:longcluster}
\end{table}

\newpage

\subsection{Large Scale Performance}
\label{sxn:large_experiment}

\subsubsection{Reddit Data}

\hfill

Reddit is a global scale social media platform organized into categories, called subreddits, in which users 
author posts. Other users are then able to author comments in response to the posts or in response to
the previous comments.  This can be represented as a comment tree, where the original post is 
the root node and any subsequent response
is a child node that is connected to the parent post or comment.

We downloaded 15 years of reddit comments from {\em pushshift.io}
\footnote{which was freely available for research at https://pushshift.io/signup}
as a dynamic multigraph edge record from the author of the child comment to the author of the parent record 
(comment or post).   
This creates a large temporal peer-to-peer network, containing 836 million vertices (author accounts) and 
7.21 billion dynamic multigraph edge records (comments).   
We seek to efficiently extract conversational bursts that are dense in time and connected.

We list several important modeling details.  The 
data is available in monthly chunks, and several accounts are deleted when the data provider acquires the data each month.   
All deleted accounts are marked with the same name, {\tt [deleted]}.
These deleted accounts could have been present in a previous month, or not, 
as it could have been created and deleted within the month.   
We represent the author of record {\tt <link\_id>} from a deleted account as {\tt [deleted]\_<link\_id>}.    
This has the effect of assigning a record from a deleted account to have a unique individual user authoring
no other records, and 
conversations where a deleted account happens to be central would be difficult to recover using topology alone. 
Nevertheless, the dataset serves our purposes, as many clusters of interest can still be extracted with our approach.

One user account {\em AutoModerator} is a known helper bot, and is the most active account in the 
dataset.   We remove the records associated with this account, which makes the computation involved quite a bit
easier.   More generally, if analyzing human-related activity is the primary goal, it 
is likely a good idea to remove other known automated accounts. However, the  {\em AutoModerator} account is the
only account we removed in our results for this paper due to sheer scale of activity.

\subsubsection{Performance}

\hfill

We measure the parallel scalability of the dendrogram building phase from Algorithm~\ref{alg:skeleton}. 
Our C++/MPI implementation is based on our team's asynchronous distributed communication library 
YGM\footnote{https://github.com/LLNL/ygm}, which has previously been utilized to improve scalability 
for several other distributed graph analytics \cite{priest2019ygm}.
The associated line graph skeleton contains 12.8 billion line graph edges, a factor of 1.78 times the 
edges in the original dynamic multigraph.

We ran all tests on the {\tt quartz} HPC system at LLNL, a large cluster of 2988
compute nodes containing 36 cores (Intel Xeon, E5-2695 v4, 2.1GHz) and 128 GB of DRAM, with a Cornelis Networks Omni-Path interconnect.
On 64, 128, 256, and 512 compute nodes 
we built the dendrogram for $\Delta t \in \{ 1,2,4, ..., 2^{28} \}$ (1 second up to over 8 years).
As can be seen in Figure~\ref{fig:reddit_strongscaling} and Table~\ref{tab:reddit_results}, we are able to compute the dendrogram on 15 years worth of reddit in under 4 minutes,
with decent strong scaling up to 256 compute nodes.

\begin{table}
\rowcolors{2}{gray!25}{white}
  \centering
  \begin{tabular}{|c||c|c|}
  \hline
  Nodes &  $\cL^*_\cG$ Construction Time (s) & Discrete Dendrogram Construction Time (s)\\
  \hline
  64 & 16.9131 & 619.584 \\
  128 & 11.453 & 340.537 \\
  256 & 9.5432 & 220.429 \\
  512 & 9.0423 & 277.24 \\
  \hline
  \end{tabular}
  \caption{Strong scaling results for line graph skeleton construction and computing the dendrogram in distributed memory via repeated 
  union-find-based connected components.  Original dynamic multigraph (author-to-author reddit comments) has 7.21B dynamic edge records and the line graph 
  skeleton has 12.8B weighted line graph edges. The approach has reasonable strong scaling until 256 compute nodes. }
  \label{tab:reddit_results}
\end{table}

\begin{figure}[h]
  \centering
  \includegraphics[width=3.0 in]{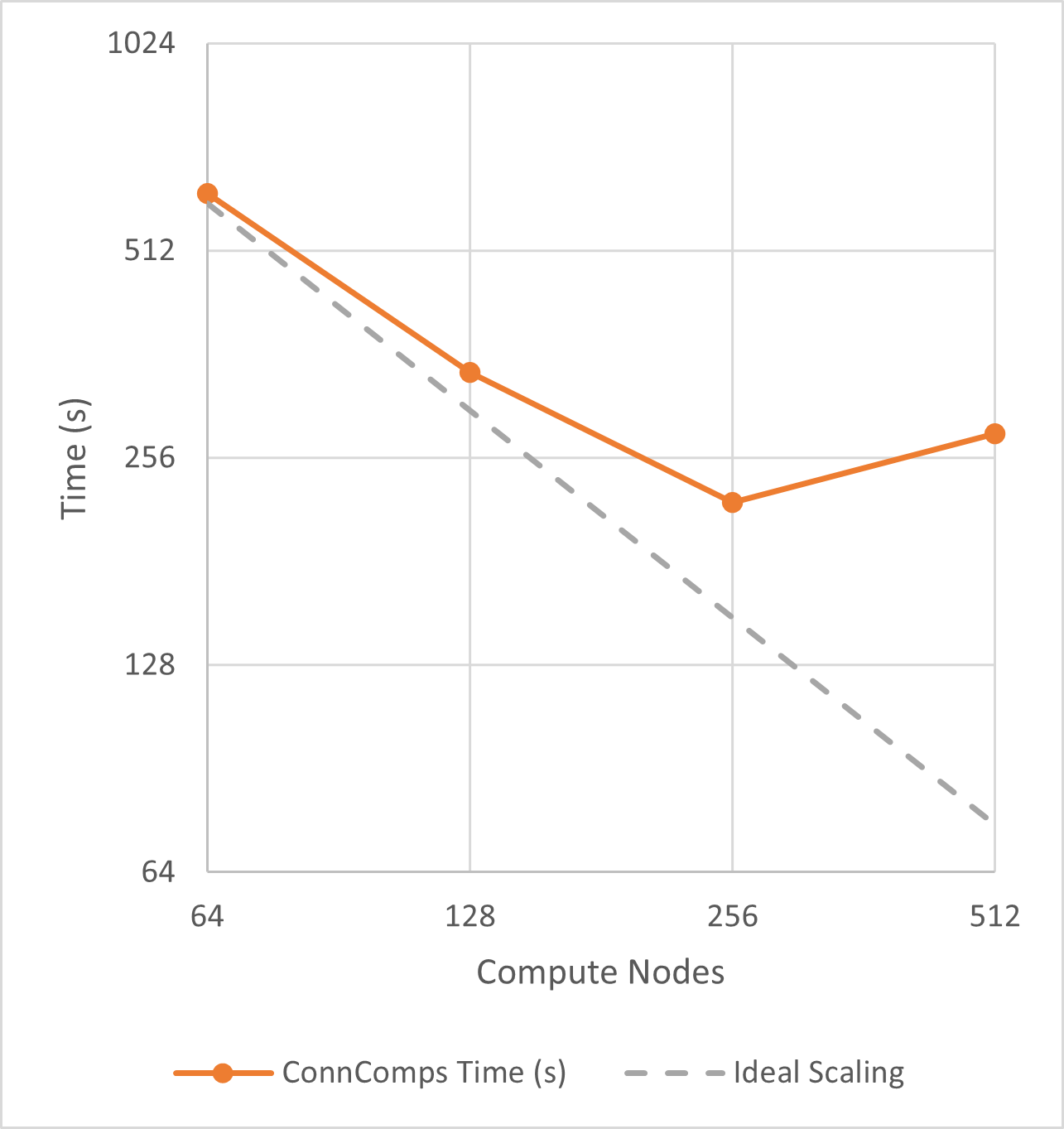}
  \caption{Strong scaling results for computing the dendrogram in distributed memory via repeated 
  union-find-based connected components.} 
  \label{fig:reddit_strongscaling} 
\end{figure}

\begin{figure}[h]
  \centering
  \includegraphics[width=5.25 in]{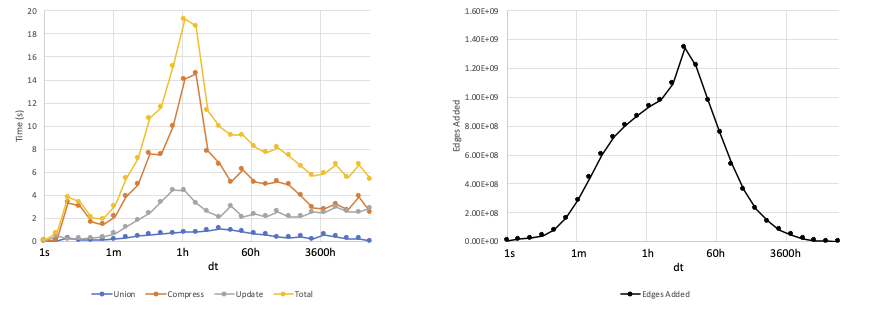}
  \caption{(Left) Timing for each phase of dendrogram construction and (Right) number of edges added, both for 256 compute nodes.} 
  \label{fig:reddit_timings} 
\end{figure}

In light of the previous section, we argue that there is much potential to apply sophisticated and 
relatively expensive topological behavior analyses to
the output clusters, which represent portions of higher-order related activity.   
The proposed clustering approach facilitates such analyses by breaking up the original dynamic multigraph 
into smaller, digestible chunks that can be related to each other by their many (dynamic) 
graph characteristics, some of which are much more easily computable on these small subgraphs than 
previous notions of clusters.   
Data analysis of the aggregated non-graph data from records within each cluster is also a possibility, 
although beyond the scope of this initial paper.

 \section{Extensions} 
\label{sxn:extensions}
 
 There are many avenues to generalize the clustering algorithm described in this paper.  
At a high level, we have presented a way to group records in relational datasets.
 These records are modeled as occurring only between two entities and at a single point in time.
 While there are many datasets where this is sufficient, it is not hard to construct
  more complex ones. Several extensions are discussed briefly with further
analysis beyond the scope of this paper.

\subsection{Nearly-Causal Relations and Time Filtered Line Graphs}

Real-world dynamic graphs often have some degree of error in time measurements (e.g. due to discretization). 
One may consider allowing an edge $f = (e_1,e_2,w_{12})$ in $\cF$ for which 
$-\delta \leq w_{12} < 0$, for some small $\delta>0$ that represents an allowable bound on the error for
 temporal observations. 
 One option to deal with {\em negative weights} in the clustering is to take the absolute value
and proceed as normal. 
As an illustrative example, consider a dataset where all timestamps have been 
 rounded to the nearest hour. For such data a $\delta =  2$ would still catch interactions between 
 pairs of edges where one timestamp was rounded down and the other up.

 Another possible adjustment is to limit the amount of time between any pair of connected records i.e. setting a 
 cutoff value $\delta_\text{max}$ such that an edge $f = (e_1,e_2,w_{12})$ is in $\cF$ 
 only if $w_{12}\leq \delta_\text{max}$. This has the potential to serve two purposes with the first being 
 a way to force the output of the clustering algorithm to return smaller 
 clusters that are more closely connected in time. The second purpose is that it reduces the total number
 of connections needed and thus the size of the resulting line graph making it easier to compute with. 
 The tradeoff with this restriction is that any interactions happening beyond the cutoff will be missed.

 \subsection{General Line Graph Weights} 
 \label{generalweights}
 
 In the bulk of our discussion we have defined the weights of the line graph to be the distance 
 in time between the corresponding edges of $\cG$, but there is nothing restricting us from
  using some other metric. For example,
   let $\cG$ be a graph with property-rich edges where each edge $e_r = (i,j)$ in $\cG$ has 
   a set of $m$ properties $U_r := \{u_{rk} : k =1,2, ... m \}$. Define a weight function that takes 
   in the properties of two edges of $\cG$ and returns a real number:
 
 $$
 f(e_1,e_2) = f(U_1,U_2) \in \mathbb{R}^+
 $$

 We can use this formulation with just timestamps for a more sophisticated weight function. 
The previous section presented one example with the idea of a cutoff time. Another example
is using a function of the time difference between records such as an exponential
function in cases where even more emphasis should be placed on short-term interactions.
 
$$
f(e_1,e_2) = e^{|t_2-t_1|^2+\|U(2,:) - U(1,:)\|^2} \qquad t_2 \geq t_1.
$$
 
Careful design of $f(\cdot, \cdot)$, or an adaptive framework that learns from examples, 
 is likely important for specific application datasets and tasks.

\subsection{Dynamic Hypergraphs} 
\label{hyperedge}
 
We consider a {\em dynamic hypergraph model}, where the hypergraph 
$\cG_h(\cV_h, \cE_h, \cT_h)$ is a discrete set of 
vertices, $\cV_h$, and a discrete set of dynamic hyperedges, 
$\cE_h$, that represent multi-way relationship events of vertices at specific times in continuous interval $\cT_h:=[0,T]$. 
Each dynamic hyperedge $e_r \in \cE_h$ is a vertex set $\cS_r$ coupled with a timestamp $t_r$ or 
$e_r = (\cS_r, t_r)$ representing a $|\cS_r|$-way 
relationship of vertices $\cS_r \subset \cV_h$ observed at time $t_r \in \cT_h$. 
For $i\in \cV_h$, let $d_i$ be the vertex degree (the number of hyperedges incident to vertex $i$) 
and for $e_r\in \cE_h$, let $d(e_r):=|\cS_r|$ be the hyperedge degree.  
This is an {\em undirected} or {\em unoriented} dynamic hypergraph model, in the sense that 
all vertices play the same role within a hyperedge.

Similar to the graph version, we can begin by defining a distance between hyperedges
based on the difference in time between when they occurred:

$$ 
\mbox{dist}( e_r, e_s) :=  |t_r - t_s|
$$
  
There exists an associated weighted line graph that is a combination of $d_i$-sized cliques and is
generally too expensive to store or work with explicitly for large graphs.   For combining the cliques 
in a multi-graph sense (versus taking the minimum distance on repeated edges), the number of multigraph 
edges in the line graph is  
$$
\sum_{i \in \cV_h} { d_i \choose 2 }
$$
This is quadratic in the size of the maximum vertex degree, and quadratic in the input for many
real-world, scale-free dynamic hypergraphs.

However, using the temporal distance $\mbox{dist}(e_r,e_s)$, much of the Algorithm~\ref{alg:edgecluster} 
immediately follows for building a clustering dendrogram with a sparse implicit skeleton representation 
of the associated line graph that faithfully represents the weighted connectivity.   
For each vertex in $\cV_h$, sort the incident hyperedges, build a local MST, 
and use the weighted connectivity of the union of all MSTs to build the dendrogram.

For an undirected dynamic hypergraph $\cG_h$, each local MST is merely a path graph that weaves its 
way through incident hyperedges in time-sorted order.

Note that the number of edges in the union of all MSTs is bounded by the average hyperedge degree 
$\overline{d}_h$ times the number of hyperedges,
$$
\sum_{i \in \cV_h} (d_i - 1) < \sum_{i \in \cV_h} d_i = \sum_{e_r \in \cE_h} d(e_r) 
= \overline{d}_h |\cE_h|.
$$
This bound is equivalent to the number of edges in the bipartite graph representation of the 
input hypergraph:     
it is the size of the input data (count of memberships of the form vertex $i$ belongs in hyperedge $e_r$), 
and is manageable no matter the average hyperedge degree size.

The methods in \S \ref{generalweights} can be applied to hypergraphs to describe a 
more general distance function between pairs of hyperedges. 
Since edges are now sets of vertices each instead of pairs, there are also additional metrics that
are unique to hypergraphs.
For example,  a Jaccard similarity score would measure how much overlap there is 
between two hyperedges allowing an additional way to measure how {\em close} one record is to another.

 \subsubsection{Directed Dynamic Hypergraphs}
 \hfill
 
For simplicity, \S\ref{hyperedge} described undirected dynamic hypergraphs.   However, a hypergraph often has 
two explicit vertex roles per hyperedge and a {\em directed dynamic hypergraph} is an 
appropriate representation.   
Examples include: a single sender broadcasts a message to multiple receivers, 
multiple reactant chemicals go through a chemical reaction that yields multiple product chemicals, and 
multiple players beat multiple other players at a team game. In these cases, Algorithm~\ref{alg:edgecluster} also yields an efficient representation of weighted line graph 
connectivity, with the local MST being formed in a similar manner to directed graphs.

\section{Conclusion}

Property graphs and, in particular, dynamic graphs are becoming ever more important in the data science landscape.
As the scale and complexity of datasets increases, it only makes sense to create new models and techniques for 
exploring the features they contain. We have presented the foundation for analyzing such datasets in two pieces;
the first is the skeleton of the line graph and the second is a distributed scale implementation of an agglomerative 
hierarchical clustering.

The full line graph generally scales to an unusable size for relatively small graphs with as few as millions of edges.
 The bounded nature of the line graph
 skeleton presented here implies that we can now gain new insights about 
 relational datasets even in the absence of a distributed compute architecture. 
 The clustering enforces a sense of causality between records 
 which yields clusters that are unrecognizable by conventional 
 graph clustering algorithms and yet still intuitively interesting. 

The clustering algorithm that we presented is not intrinsically tied to computing clusters specifically for the line graph.
In general, any weighted graph could be clustered using Algorithm \ref{alg:hier_cluster} with large scale graphs additionally
 using the discrete dendrogram presented in \S \ref{sec:distributed}.
 The union of the line graph skeleton and this hierarchical graph clustering provides an often sought after 
 win-win in algorithm development as the clustering only needs the line graph 
 skeleton and that is all we can reasonably compute 
 for most graphs anyway. The potential modularity and flexibility of this approach should generate significant future 
research interest.

\appendix

\section{Equivalent Weighted Connectivity of Line Graph Skeleton} 
\label{sxn:theoryappendix} 
 
In this section, we demonstrate  theory that shows 
Algorithm~\ref{alg:skeleton} will produce a subgraph of the full 
time incremented line graph 
that faithfully represents the weighted connectivity. 
 
Let $\cG_s(\cV, \cE)$ be a static, unweighted graph. 
We say $\cG_s$ is {\em connected}, if for any vertex pair $i,j \in \cV$, there exists at least one connected path of edges from $i$ to $j$ in $\cE$. 
If $\cG_s$ is disconnected, it is decomposable into a set of {\em connected components}, $\comps{\cG_s}$, a set of subgraphs that covers all vertices and edges. 
 
For weighted graphs consider computing connected components given a weight threshold. 
In this work, we use weights to represent distances, i.e. a smaller weight on an edge implies a stronger relationship 
because the connected nodes are closer together. 
Let $\cG(\cV, \cE, \cW)$ be a weighted graph, 
and let $\cG_\omega(\cV, \cE_\omega)$ be the {\em weight-filtered graph} formed by thresholding edge weights 
$$ 
(i,j,w_{ij}) \in \cE_\omega 
\qquad \Longleftrightarrow \qquad 
w_{ij} \leq \omega. 
 $$ 
 
\begin{definition}{\sc (Weight-Filtered Connected Components)} 
The {\em weight-filtered connected components} of a weighted graph $\cG(\cV, \cE, \cW)$ with threshold $\omega$ is 
$$ 
\comps{\cG, \omega} := \comps{\cG_\omega}. 
$$ 
If two weighted graphs $\cG_1$ and $\cG_2$ have the same vertex set and:
$$ 
\comps{\cG_1, \omega} = \comps{\cG_2, \omega} 
$$ 
for all $\omega$, we say their weighted connectivities are {\it equivalent}. 
\end{definition}

\begin{definition}{\sc (MST)}. 
\label{thm:mwst} 
For a weighted graph $\cG(\cV, \cE, \cW)$, a {\em weighted spanning tree} $\cM(\cV,\cE_\cM, \cW_\cM)$ is a connected tree involving all vertices $\cV$. 
The {\em weight} of $\cM$ is the sum of all the weights, $\sum_{(i,j)\in \cE_\cM} w_{ij}$. 
Additionally, if no spanning tree exists with lower weight, then $\cM$ is a {\em minimum weight spanning tree (MST)}. 
 
 \end{definition}

\begin{theorem}{\sc (Weighted MST Connectivity)}. 
\label{thm:mwstconn}  
Let $\cG(\cV,\cE,\cW)$ be a connected and weighted graph 
and $\cM(\cV, \cE_\cM, \cW_\cM)$ be any MST 
of $\cG$.   The weighted connectivities of $\cG$ and $\cM$ are equivalent. 
\end{theorem} 
 
\begin{proof} 
We prove via contradiction. 
Assume there exists a threshold $\omega$ for which $\comps{\cM, \omega} \neq \comps{\cG, \omega}$. 
Because $\cE_\cM \subset \cE$, this implies there is a pair of vertices $r,s$ that are connected in a path $\cP$ within $\cE$ whose edge weights are all less than or equal to $\omega$, 
but no such path exists in $\cE_\cM$.   

Let $\cP_0$ be the unique path from $r$ to $s$ within tree $\cE_{\cM}$, which contains at least one edge $(i,j)$ with weight above the threshold, $w_{ij} \geq w$.
Remove $(i,j)$ from $\cM$ to yield two subtrees, $\cM_i$ and $\cM_j$, where 
$\cV_i$ and $\cV_j$ are the disjoint 
associated vertex subsets connected by these trees and 
$\cE_i$ and $\cE_j$ are the associated edge subsets. 
Note that $\cV_i \cup \cV_j = \cV$ and no edge exists between 
$\cV_i$ and $\cV_j$.   
Without loss of generality let $i,r \in \cV_r$ and $j,s \in \cV_s$.

The existence of $\cP \in \cE$ shows there must exist an edge 
$(i^\star, j^\star)$ with $w_{i^\star, j^\star} \leq \omega$, 
$i^\star \in \cV_i$, 
and  
$j^\star \in \cV_j$. 
Letting $\cE_{\cM^\star} = \{\cE_{\cM} \setminus (i,j)\} \cup \{(i^\star, j^\star)\}$, yields a spanning tree with lower weight
than $\cM$ and the contradiction has been realized.

\end{proof} 

We show Kruskal's Algorithm \cite{kruskal1956}  
applied to an explicitly formed local line graph will give the same MST as the 
implicitly constructed tree in Algorithm~\ref{alg:skeleton}, where no explicit line graph is ever built.
For a general weighted graph, there is no guarantee for an MST to be unique, so 
we perturb the edge weights by a sufficiently small amount so each (line graph) 
edge weight is unique, and the approach we take also gives us a deterministic (line graph) 
vertex ordering.   

\begin{theorem}{\sc (Causal time increment MST)}.
\label{thm:1dmwstd}
Let $\cL_\cG(i) = (\cV,\cE,\cW)$ be a node-local line graph representing a local portion of a line graph, 
from Definition \ref{def:nliwlg}. 
The weights are defined as time increments: 
\begin{itemize}
\item[(i)] every $p \in \cV$ has an associated value $t_p \in \mathbb{R}$, 
\item[(ii)] $(p,q) \in \cE$ only if $t_p < t_q$,
and 
\item[(iii)] for every $(p,q) \in \cE$, we have $w_{pq} = t_q - t_p \geq 0$.
\end{itemize}
The $\cE_\cM$ defined by Algorithm~\ref{alg:skeleton} 
yields a tree $\cM(\cV, \cE_\cM, \cW)$ that is an MST of $\cL_\cG(i)$.
\end{theorem}

\begin{proof}
We break ties in edge weights in a specific way by perturbing edge weights 
by a small enough amount that ensures the edges of any MST of the perturbed graph is also
an MST for the original.   Let $\gamma = w_{min} / (2 |\cE_i|^2) $, where $w_{min}$ is the smallest
weight in $\cL_\cG(i)$.   For each incoming edge $e_p = (i_p, i, t_p)$, 
replace it with $e'_p = (i_p, i, t'_p)$ and for each outgoing edge $e_q = (i, j_q, t_q)$, 
replace it with $e'_q = (i, j_q, t'_q)$, where 

\begin{eqnarray*}
t'_p & = & t_p + \gamma^2 p \\
t'_q & = & t_q + \gamma q. 
\end{eqnarray*}
The choice of perturbing by small values of $\gamma$ and $\gamma^2$ allows the edge weights
to be unique, positive, and represent a deterministic vertex sorting that allows the two 
algorithms to compute the exact same MST. 

Define the weighted rectangular matrix $B$ whose rows are the incoming edges 
(line graph vertices) incident to $i$.   Respectively, columns of $B$ are the outgoing edges 
(line graph vertices) incident to vertex $i$.  Sort the rows and columns by $t'_p$.  
Entries of $B$ are the weights, 

$$
B_{pq} = \left\{
\begin{array}{rcl}
t'_q - t'_p && t'_q \geq t'_p \\
+\infty && t'_q < t'_p   
\end{array}
\right\},
$$

where non-edges are chosen to be infinity (in the context of computing an MST).   The 
non-infinite entries of $B$ are 

$$
B_{pq} = t_q - t_p + \gamma q - \gamma^2 p.
$$

Kruskal's algorithm is greedy.   It starts with an empty set and iteratively 
adds the smallest available 
edge weight that does not introduce a cycle to the set, building up the MST.
For a general (line graph) edge $(p,q)$ with weight $B_{pq}$, it is easy to show the following relations

\begin{eqnarray*}
B_{pq} > B_{p,q-1} && \mbox{ if } \, B_{p,q-1} \neq +\infty,  \\
B_{pq} > B_{p+1,q} &&  \mbox{ if } \, B_{p+1,q} \neq +\infty.  \\
\end{eqnarray*}

These show that the first edge selected by Kruskal's Algorithm must be one of the following cases: 

(i) $p=nrows(B)$ and $q=1$, 

(ii)  $p=nrows(B)$ and $B_{p,q-1}=\infty$, 

(iii) $B_{p+1,q} = \infty$ and $q=1$, or

(iv) $B_{p,q-1} = B_{p+1,q} = \infty$.

In each of the cases (i)-(iv), the edge $(p,q)$ would also be selected by Algorithm~\ref{alg:skeleton}  
(e.g. an edge associated with one of the orange highlighted ``X" entries in Figure~\ref{fig:skeleton}).

\begin{figure}[h]
\begin{center}
\includegraphics[width=3.4 in]{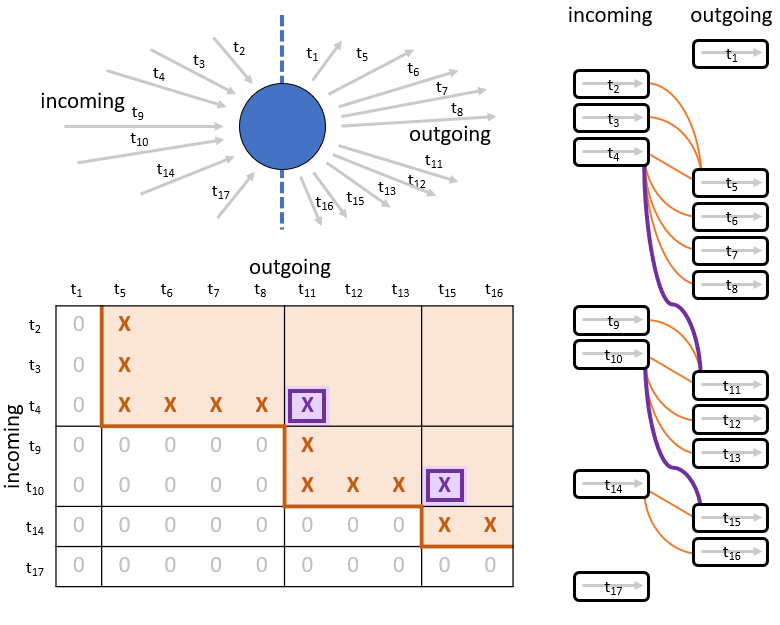} 
\caption{
{\sc Node Local Construction of Line Graph Skeleton.}
\\\hspace{\textwidth}
(Top-Left) A example of a single vertex that has 
several incoming and outgoing edges with 
observation times $t_p$.
\\\hspace{\textwidth}
(Right) Sorting edges by time (top to bottom) and performing 
Alg.~\ref{alg:nlmst} 
yields a tree between the incoming and outgoing edges. 
\\\hspace{\textwidth}
(Bottom-Left) A matrix representation of the 
$\cL_\cG(i)$, with incoming and outgoing edges represented by  
 rows and columns respectively. The orange shaded 
entries are all line graph edges that exist in the $\cL_\cG$ 
(which is never explicitly formed in our approach). 
The ``X's'' represent 
edges added to the line graph skeleton $\cL^*_\cG$ during 
Algorithm~\ref{alg:nlmst} and correspond 
to same color edges in the graph representation on the right.
}
\label{fig:skeleton}
\end{center}
\end{figure}

In subsequent iterations of Kruskal's Algorithm, the selected edge will have one of:

(i)-(iv) from above, 

(ii') $p=nrows(B)$ and $(p,q-1)$ was previously selected, 

(iii') $(p+1,q)$ was previously selected and $q=1$, or 

(iv') $(p+1,q)$ was previously selected and $(p,q-1)$ was previously selected.

For cases (i)-(iv), (ii') and (iii'), $(p,q)$ would also selected by Algorithm~\ref{alg:skeleton}.

Case (iv') is more complex, 
such edges that do not induce a cycle would be selected by Kruskal's.   
We see that the only possibility is that edges for which $B_{p+1,q-1} = \infty$ 
(e.g. the purple highlighted ``X" entries in Figure~\ref{fig:skeleton}) 
are the only possibility. 
Such edges are explicitly chosen in Algorithm~\ref{alg:skeleton}.

To see the alternative is not possible, 
if  $B_{p+1,q-1} < \infty$, then
$B_{p,q} > B_{p,q-1} > B_{p+1,q-1}$ and edge $(p+1, q-1)$ would have been also selected
3 or more iterations earlier.   The four edges $(p,q), (p+1,q), (p+1,q-1), (p,q-1)$ represent a 
four cycle and $(p,q)$ would not be selected by Kruskal's algorithm.   
Such edges are also not selected in Algorithm~\ref{alg:skeleton}.

Repeating this shows Kruskal's algorithm builds up the same tree
as Algorithm~\ref{alg:skeleton}, implying the output is an MST of 
the graph associated with the perturbed $B$.   
This is also MST of $\cL_\cG(i)$, by construction 
(we chose the perturbation of edge weights to ensure this property).

\end{proof}

\begin{theorem}{\sc (Global Connectivity)}
\label{thm:lasttheorem}
Let $\cG(\cV,\cE,\cT)$ be a dynamic multigraph and let $\cL_\cG$ be the
associated time-increment weighted line graph. 

Define the skeleton 

$$
\cL^*_\cG := \bigcup_{i \in \cV} \cM_\cG(i)
$$

as produced by Algorithm~\ref{alg:skeleton}.   
The global weighted connectivities of $\cL_\cG$ and $\cL^*_\cG$ are equivalent. 
\end{theorem}
\begin{proof}
For a given $\omega$, 
let $\cP$ be a path in $\cL_\cG$ between any source/target vertex pair in
$\cL_\cG$ such that all edge weights are less than or equal to $\omega$.
By Theorems~\ref{thm:mwstconn} and \ref{thm:1dmwstd}, 
each edge $(e_p:=(i_p, i, t_p),e_q:=(i, j_q, t_q), w_{pq}:= |t_p-t_q|) \in \cP$
yields a path $\cP_i$ in $\cM_\cG(i)$ 
such that all weights in $\cP_i$ are also less than $\omega$.   
The union of all $|\cP|$ paths yields a path $\cP^*$ in
$\cL^*_\cG$ 
connecting the same source/target vertex pair 
with edges all weight $\omega$ or less.  Note these constituent paths 
could stem from the same center vertex $i$ and associated 
bipartite graph $\cM_\cG(i)$ multiple times 
(and often would in real-world cases).

\end{proof}

\bibliographystyle{comnet} 
\bibliography{lineg}

\end{document}